\documentclass[11pt]{article}

\usepackage{fullpage}
\usepackage{amssymb}
\usepackage{amsmath}
\usepackage{amsthm}
\usepackage{txfonts}

\newcommand{\CCC}{\mathcal{C}} 
 
\newcommand{\GGG}{\mathcal{G}} 
 
 \newcommand{\LLL}{\mathcal{L}}

\newcommand{\SSS}{\mathcal{S}}

\newtheorem{theorem}{Theorem}
\newtheorem{definition}[theorem]{Definition}

\newtheorem{corollary}[theorem]{Corollary}
\newtheorem{lemma}[theorem]{Lemma}

\newcommand{\Am}{\mathbb{A}}
\newcommand{\Bm}{\mathbb{B}}
\newcommand{\Aa}{\Am}
\newcommand{\Bb}{\Bm}

\newcommand{\nats}{\varmathbb{N}}

\newcommand{\mb}{\mathbf}

\newcommand{\ra}{\rightarrow}

\newcommand{\pfunc}{\rightharpoonup}

\newcommand{\minor}{\preceq}
\newcommand{\grad}{\nabla}
\newcommand{\iso}{\cong}
\newcommand{\defeq}{:=}

\renewcommand{\hat}{\widehat}

\newcommand{\tup}[1]{\mathbf{#1}}

\renewcommand{\phi}{\varphi}

 \title{Homomorphism Preservation on Quasi-Wide Classes}
 \author{Anuj Dawar \\University of Cambridge Computer Laboratory, UK}

\begin{document}


\maketitle

\begin{abstract}
  A class of structures is said to have the homomorphism-preservation
  property just in case every first-order formula that is preserved by
  homomorphisms on this class is equivalent to an existential-positive
  formula.  It is known by a result of Rossman that the class of
  finite structures has this property and by previous work of Atserias
  et al.\ that various of its subclasses do.  We extend the latter
  results by introducing the notion of a quasi-wide class and showing
  that any quasi-wide class that is closed under taking substructures
  and disjoint unions has the homomorphism-preservation property.  We
  show, in particular, that classes of structures of bounded expansion
  and classes that locally exclude minors are quasi-wide.  We also construct an
  example of a class of finite structures which is closed under
  substructures and disjoint unions but does not admit the
  homomorphism-preservation property.
\end{abstract}

\section{Introduction}\label{sec:intro}

Preservation theorems are model-theoretic results that link semantic
restrictions on a logic with corresponding syntactic restrictions.
For instance, the {\L}o\'s-Tarski preservation theorem guarantees that
any first-order formula whose models are closed under extensions
is equivalent to an existential formula.  In the early development of
finite model theory, it was noted that many classical preservation
theorems of model theory fail when we are only interested in finite
structures (see~\cite{Gur84}).  The {\L}o\'s-Tarski theorem is an
example of one such---it was noted by Tait~\cite{Tai59} that there
are formulas of first-order logic whose \emph{finite} models are
closed under extension but that are not equivalent, even in
restriction to finite structures, to an existential formula.
Similarly, Ajtai and Gurevich~\cite{AG87} established that Lyndon's
theorem---which implies that any formula that is monotone on all
structures is equivalent to one that is positive---also fails in the
finite.  One example of a preservation theorem whose status in the
finite remained open for many years is the homomorphism preservation
theorem.  This states that a first-order formula whose models are
closed under homomorphisms is equivalent to an existential-positive
formula.  Rossman recently proved~\cite{Ros08} that this holds, even
when we restrict ourselves to finite structures.

A recent trend in finite model theory has sought to examine
model-theoretic questions, such as the preservation properties, not
just on the class of all finite structures but on subclasses that are
of interest from the algorithmic point of view (see~\cite{Daw07-mfcs} for an
overview of results in this direction).  Thus, prior to Rossman's
result, Atserias et al.~\cite{ADK06} proved that the homomorphism
preservation theorem holds in any class of structures $\CCC$ of
\emph{bounded treewidth} which is closed under substructures and
disjoint unions.  More generally, they showed that homomorphism
preservation holds on $\CCC$ provided that the \emph{Gaifman graphs}
of structures in $\CCC$ exclude some minor and $\CCC$ is closed under
substructures and disjoint unions.  Note that these results are not
implied by Rossman's theorem.  Indeed, if we consider two classes
$\CCC \subseteq \CCC'$, we cannot conclude anything about whether or
not homomorphism preservation holds on $\CCC$ from the fact that it
holds on $\CCC'$.  An example of a class of finite structures on which
homomorphism preservation fails is discussed in
Section~\ref{sec:counterexample}.

An open question that was posed in~\cite{ADK06} was whether the
results from that paper could be extended to other classes, in
particular by replacing the requirement that $\CCC$ exclude a minor by
the requirement that $\CCC$ have \emph{bounded local treewidth} as
defined in~\cite{Epp00,FG01}.  This restriction is incomparable with
the requirement that $\CCC$ excludes a minor, in the sense that there
are classes with an excluded minor that do not have bounded local
treewidth and vice versa.  However, there is a common generalisation
of the two in the notion of \emph{locally excluded minors} introduced
by Dawar et al.~\cite{DGK07}.  In this paper, we answer the open
question from~\cite{ADK06} by showing that any class $\CCC$ of finite
structures that locally excludes a minor and is closed under taking
substructures and disjoint unions satisfies the homomorphism
preservation property.  We also establish this for classes of
\emph{bounded expansion}, as defined by Ne\v{s}et\v{r}il and Ossona de
Mendez~\cite{NOdM05}.

The proof given in~\cite{ADK06} that classes of structures that
exclude a minor satisfy homomorphism preservation was composed of two
elements.  First, a result derived from a lemma by Ajtai and
Gurevich~\cite{AG94} that showed a certain density property for
minimal models of a formula $\phi$ that is preserved under
homomorphisms.  This implies that if a class $\CCC$ satisfies the
condition of being \emph{almost wide} (this is defined in
Section~\ref{sec:prelim} below) and is closed under substructures and
disjoint unions, then $\CCC$ satisfies homomorphism preservation.
Secondly, we showed, using a combinatorial construction
from~\cite{KS99}, that any class that excludes some graph as a minor is
almost wide.  In order to extend these results
to classes that locally exclude a minor and classes of bounded
expansion, we first define a relaxation of the almost wideness
condition to one we term \emph{quasi-wideness}.  We show that the
Ajtai-Gurevich lemma can be adapted to show that any class $\CCC$
which is quasi-wide and closed under substructures and disjoint unions
also satisfies homomorphism preservation.  This is established in
Section~\ref{sec:ajtai-g}.  Then, an extension of the
combinatorial argument from~\cite{ADK06} establishes that classes of
bounded expansion and classes that locally exclude a minor are
\emph{quasi-wide}.  These arguments are presented in
Section~\ref{sec:classes}.

The steady recurrence of the requirement that $\CCC$ is closed under
substructures and disjoint unions arises from the fact that these are
the constructions used in the density argument of Ajtai and Gurevich.
A natural question that arises is 
whether these conditions alone might be sufficient to guarantee
homomorphism preservation.  However, this is not the case, as we
establish through a counter-example constructed in
Section~\ref{sec:counterexample}.

I announced the results presented here in an invited lecture~\cite{Daw07-mfcs}, without presenting the proofs.
Since then, Ne\v{s}et\v{r}il and
Ossona de Mendez have extended the combinatorial argument from
Section~\ref{sec:classes} and provided an elegant characterisation of
quasi-wide classes that are closed under substructures~\cite{NOdM08}.

\textbf{Acknowledgements:} The results reported here were obtained
during a visit made to Cambridge by Guillaume Malod in the summer of
2007.  I am grateful to him for stimulating discussions and for his
help with the material.  I am also grateful to Jarik Ne\v{s}et\v{r}il
for his repeated encouragement to write this paper ever since I told
him the results.

\section{Preliminaries}\label{sec:prelim}
This section contains the definitions of some basic notions
and a minimum amount of background material.

\subsection{Relational Structures }

A \emph{relational vocabulary} $\sigma$ is a finite set  of
\emph{relation symbols}, each with a specified \emph{arity}. A
\emph{$\sigma$-structure} ${\Aa}$ consists of a \emph{universe}
$A$, or \emph{domain}, and an \emph{interpretation} which associates
to each relation symbol $R\in \sigma$ of some arity $r$, a relation
$R^{\Aa}\subseteq A^r$. A \emph{graph} is a structure ${\bf G} =
(V,E)$, where  $E$ is a binary relation that is symmetric and
irreflexive. Thus, our graphs are undirected, loopless, and without
parallel edges.

A $\sigma$-structure ${\Bb}$ is called a \emph{substructure} of
${\Aa}$ (and we write $\Bb \subseteq \Aa$) if $B\subseteq A$ and
$R^{\Bb}\subseteq R^{\Aa}$ for every $R\in \sigma$. It is called an
\emph{induced substructure} if $R^{\Bb} = R^{\Aa} \cap B^r$ for every
$R\in \sigma$ of arity $r$.  Note that this terminology is at variance
with common usage in model theory where the term ``substructure'' is
used for what we call an ``induced substructure''.  However, it is
more convenient for us as, for the purpose of studying properties
preserved under homomorphisms, we are more interested in substructures
that are not necessarily induced.  Note also the analogy with the
concepts of \emph{subgraph} and \emph{induced subgraph} from graph
theory. A substructure ${\Bb}$ of ${\Aa}$ is proper if
${\Aa}\not={\Bb}$.

A \emph{homomorphism} from ${\Aa}$ to ${\Bb}$ is a mapping
$h:A\rightarrow B$ from the universe of ${\Aa}$ to the universe of
${\Bb}$ that preserves the relations, that is if
$(a_1,\ldots,a_r)\in R^{\Aa}$, then $(h(a_1),\ldots,h(a_r))\in
R^{\Bb}$.  We say that two structures $\Aa$ and $\Bb$ are
{\em homomorphically equivalent\/} if there is a homomorphism from
$\Aa$ to $\Bb$ and a homomorphism from $\Bb$ to $\Aa$.
Note that, if $\Aa$ is a substructure of $\Bb$, then the
injection mapping is a homomorphism from $\Aa$ to $\Bb$.  If the
homomorphism $h$ is bijective and its inverse is a homomorphism from
$\Bb$ to $\Aa$ then $\Aa$ and $\Bb$ are isomorphic and we write $\Aa
\iso \Bb$.

For a pair of structures $\Aa$ and $\Bb$, we write $\Aa \oplus \Bb$
for the \emph{disjoint union} of $\Aa$ and $\Bb$.  That is, $\Aa
\oplus \Bb$ is the structure whose universe is the disjoint union of
$A$ and $B$ and where, for any relation symbol $R$ and any tuple of
elements $\mb{t}$, we have $\mb{t} \in R^{\Aa \oplus \Bb}$ just in
case either $\mb{t} \in R^{\Aa}$ or $\mb{t} \in R^{\Bb}$.

The \emph{Gaifman graph} of a $\sigma$-structure ${\Aa}$, denoted
by $\mathcal{G}({\Aa})$, is the (undirected) graph whose set of nodes
is the universe of ${\Aa}$, and whose set of edges consists of all
pairs $(a,a')$ of distinct elements  of $A$ such that $a$ and $a'$
appear together in some tuple of a relation in ${\Aa}$. 

Let $\mb{G} = (V,E)$ be a graph.  Recall that the \emph{distance} between two vertices $u$ and $v$ is the length of the shortest path from $u$ to $v$.  For a vertex $u$ and an integer $r\geq 0$,  \emph{$r$-neighborhood} of $u$
in $\mb{G}$, denoted by $N_r^\mb{G}(u)$, is the set of vertices at distance atmost $r$ from $u$.  In particular, $N_0^\mb{G}(u) = \{u\}$.

Where this causes no confusion, we also write $N_r^\mb{G}(u)$ for the subgraph of $\mb{G}$ induced by this set of vertices.
Similarly, for a structure $\Aa$ and an element $a$ in its universe, we write
$N_r^{\Aa}(a)$ both for the set $N_r^{\GGG(\Aa)}(a)$ and the substructure of $\Aa$ it induces.

\subsection{Logic}

Let $\sigma$ be a relational vocabulary.  The \emph{atomic formulas}
of $\sigma$ are those of the form $R(x_1,\ldots,x_r)$, where
$R\in\sigma$ is a relation symbol of arity $r$, and $x_1,\ldots,x_r$
are first-order variables that are not necessarily distinct.  Formulas
of the form $x=y$ are also atomic formulas, and we refer to them as
\emph{equalities}. The collection of \emph{first-order formulas} is
obtained by closing the atomic formulas under negation, conjunction,
disjunction, universal and existential first-order quantification. The
semantics of first-order logic is standard.  If ${\Aa}$ is a
$\sigma$-structure and $\varphi$ is a first-order formula, we use the
notation ${\Aa} \models \varphi[\mb{a}]$ to denote the fact that
$\varphi$ is true in ${\Aa}$ when its free variables are interpreted
by the tuple of elements $\mb{a}$.  When $\phi$ is a sentence (i.e.\
contains no free variables), we simply write ${\Aa} \models \varphi$.
The collection of \emph{existential-positive} first-order formulas is
obtained by closing the atomic formulas under conjunction,
disjunction, and existential quantification. By substituting
variables, it is easy to see that equalities can be eliminated from
existential-positive formulas.

We say that a first-order formula $\phi$ is \emph{preserved under
  homomorphisms} if, whenever $\Aa \models \phi[\mb{a}]$ and
$h:A\rightarrow B$ is a homomorphism from $\Aa$ to $\Bb$ then $\Bb
\models \phi[h({\mb{a}})]$.  It is an easy exercise to show that any
existential positive first-order formula is preserved under
homomorphisms.  The homomorphism preservation theorem provides a kind
of converse to this statement: every first-order formula that is
preserved under homomorphisms is logically equivalent to an
existential positive formula.

We are interested in versions of homomorphism preservation on
restricted classes of structures.  If $\CCC$ is a class of structures,
we say that a formula $\phi$ is \emph{preserved under homomorphisms on
  $\CCC$} if whenever $\Aa$ and $\Bb$ are structures in $\CCC$, $\Aa
\models \phi[\mb{a}]$ and $h:A\rightarrow B$ is a homomorphism from
$\Aa$ to $\Bb$ then $\Bb \models \phi[h({\mb{a}})]$.  We say that two
formulas $\phi$ and $\psi$ are \emph{equivalent on $\CCC$} if for every
structure $\Aa$ in $\CCC$ we have $\Aa \models (\phi \leftrightarrow
\psi)$.  We say that $\CCC$ has the \emph{homomorphism preservation
  property} if every formula $\phi$ that is preserved under
homomorphisms on $\CCC$ is equivalent on $\CCC$ to an existential
positive formula.  By a theorem of Rossman~\cite{Ros08}, the class of
finite structures has the homomorphism preservation property.

For a sentence $\phi$ preserved under homomorphisms on a class of
structures $\CCC$, we say that $\Aa \in \CCC$ is a \emph{minimal
  model} of $\phi$ in $\CCC$ if $\Aa\models\phi$ and for every proper
substructure $\Bb \subseteq \Aa$ such that $\Bb\in \CCC$, $\Bb
\not\models \phi$.  The following lemma is established by an easy
argument sketched in~\cite{ADK06}.
\begin{lemma}\label{lem:minimal}
  Let $\CCC$ be a class of finite structures closed under taking
  substructures and let $\phi$ be a sentence that is preserved under
  homomorphisms on $\CCC$.  Then the following are equivalent:
  \begin{enumerate}
  \item $\phi$ has finitely many minimal models in $\CCC$.
  \item $\phi$ is equivalent on $\CCC$ to an existential-positive
    sentence. 
  \end{enumerate}
\end{lemma}
The main consequence of this lemma is that in order to establish that
$\CCC$ has the homomorphism preservation property, it suffices to
establish an upper bound on the size of the minimal models.  To be
precise, we aim to prove that for any $\phi$ there is an $N$ such that
no minimal model of $\phi$ is larger than $N$.

The quantifier rank of a first-order formula $\phi$ is just the
maximal depth of nesting of quantifiers in $\phi$.
For every integer $r \geq 0$, let $\delta(x,y) \leq r$ denote 
the first-order formula expressing that the distance
between $x$ and $y$ in the Gaifman graph is at most $r$.
Let $\delta(x,y) > r$ denote the negation of this formula.
Note that the quantifier rank of $\delta(x,y) \leq r$ is
bounded by $r$. A \emph{basic local sentence} is a sentence of
the form
\begin{equation}
\exists x_1 \cdots \exists x_n
\left({\bigwedge_{i\not=j} \delta(x_i,x_j) > 2r
\wedge \bigwedge_i \psi^{N_r(x_i)}(x_i)}\right),
\label{equation:basiclocal}
\end{equation}
where $\psi$ is a first-order formula with one
free variable. Here, $\psi^{N_r(x_i)}(x_i)$ stands
for the relativization of $\psi$ to $N_r(x_i)$; that is, 
the subformulas of $\psi$ of the form $\exists x \theta$ are 
replaced by $\exists x(\delta(x,x_i) \leq r \wedge \theta)$,
and the subformulas of the form $\forall x \theta$ are
replaced by $\forall x (\delta(x,x_i) \leq r \rightarrow \theta)$.
The \emph{locality radius} of a basic local sentence
is $r$. Its \emph{width} is $n$.  The formula $\psi$ is called the
\emph{local condition}.

The main value of basic local sentences is that they form a
building block for first-order logic.  This follows from Gaifman's
Theorem (for a proof, see, for example, \cite[Theorem~2.5.1]{EF99}),
which states that every first-order sentence is equivalent to a
Boolean combination of basic local sentences.  We will need a refined
version of this, which takes account of quantifier rank.  The
following statement follows immediately from the proof given
in~\cite{EF99}. 
\begin{theorem}[Gaifman]\label{thm:gaifman}
  Every first-order sentence $\phi$ of quantifier rank at most $q$ is
  equivalent to a Boolean combination of basic local sentences of
  locality radius at most $7^q$.
\end{theorem}
Indeed, a better bound than $7^q$ on the locality radius is possible, but the exact value of the bound will not concern us here.  
It is important, however, that the upper bound  does
not depend on the signature $\sigma$.

\subsection{Graphs}

We are interested in classes of finite structures $\CCC$ defined by a
graph-theoretic restriction on their Gaifman graphs.  In order to
define these restrictions, we introduce some graph theoretic
concepts.  For further details on graph minors, the reader is referred
to~\cite{Die05}.  For a graph $\mb{G}$, we often write $V^{\mb{G}}$
for the set of its vertices and $E^{\mb{G}}$ for the set of its
edges.  For $A \subseteq V^{\mb{G}}$, we write $\mb{G}[A]$ to denote
the subgraph of $\mb{G}$ induced by the set of vertices $A$.

We say that a graph ${\mb{G}}$ is a \emph{minor} of $\mb{H}$ (written
$\mb{G} \minor \mb{H}$) if
${\mb{G}}$ can be obtained from a subgraph of $\mb{H}$ by
contracting edges. The contraction of an edge $(u,v)$ consists in
replacing its two endpoints with a new vertex $w$ whose neighbours are all nodes that were neighbours of either $u$ or $v$.  An equivalent characterization (see \cite{Die05})
states that $\mb{G}$ is a minor of $\mb{H}$ if there is a map that
associates to each vertex $v$ of $\mb{G}$ a non-empty
\emph{connected} subgraph $\mb{H}_v$ of $\mb{H}$ such that
$\mb{H}_u$ and $\mb{H}_v$ are disjoint for $u \neq v$ and if there
is an edge between $u$ and $v$ in $\mb{G}$ then there is an edge in
$\mb{H}$ between some node in $\mb{H}_u$ and some node in
$\mb{H}_v$.  The subgraphs $\mb{H}_v$ are called \emph{branch sets}. 

We say that a class $\CCC$ of finite graphs \emph{excludes $\mb{G}$
as a minor} if, for every $\mb{H}$ in $\CCC$, $\mb{G} \not\minor
\mb{H}$.  We say that $\CCC$ \emph{excludes a minor} if there is some
graph $\mb{G}$ such that $\CCC$ excludes $\mb{G}$ as a minor.  Note
that if $\mb{G}$ is a graph on $n$ vertices and $\mb{K}_n$ is the
clique on $n$ vertices, then $\mb{G} \minor \mb{K}_n$.  Thus, if
$\CCC$ excludes a minor, then there is an $n$ such that $\CCC$
excludes $\mb{K}_n$ as a minor.  Among classes of graphs that exclude
a minor are the class of planar graphs, or more generally, the class
of graphs embeddable into any given fixed surface.

The notion of graph classes with locally excluded minors is introduced
in~\cite{DGK07}.  We say that a class $\CCC$ \emph{locally excludes
  minors} if there is a function $f: \nats \ra \nats$ such that for
each $\mb{G}$ in $\CCC$ and each vertex $v$ in $\mb{G}$,
$\mb{K}_{f(r)} \not\minor N_r^{\mb{G}}(v)$.  That is, for every $r$,
the class of graphs $\CCC_r$, formed from $\CCC$ by taking the
neighbourhoods of radius $r$ around all vertices of graphs in $\CCC$,
excludes a minor.

Finally, we define classes of bounded expansion, as introduced by
Ne\v{s}et\v{r}il and Ossona de Mendez~\cite{NOdM05}.   We say that
$\mb{G}$ is a \emph{minor at depth $r$} of $\mb{H}$ (and write $\mb{G}
\minor_r \mb{H}$)  if $\mb{G}\minor \mb{H}$ and this is witnessed 
by a collection of
branch sets $\{ \mb{H}_v \mid v \in V^{\mb{G}} \}$, 
each of which is contained in a neighbourhood of $\mb{H}$ of radius
$r$.  That is, for each $v \in V^{\mb{G}}$, there is a $w \in
V^{\mb{H}}$ such that $\mb{H}_v \subseteq N_r^{\mb{H}}(w)$.  For any
graph $\mb{H}$, the \emph{greatest reduced average density} (or
\emph{grad}) \emph{of radius $r$} of $\mb{H}$, written
$\grad_r(\mb{H})$  is defined as
$$ \grad_r(\mb{H}) = \max \big\{ \frac{|E^{\mb{G}}|}{|V^{\mb{G}}|}
\mid \mb{G} \minor_r \mb{H} \big\}.$$
In other words, $\grad_r(\mb{H})$ is half the maximum average degree
that occurs among minors of $\mb{H}$ of depth $r$.  In particular, if
$d(\mb{G})$ denotes the average degree of $\mb{G}$, then
$\grad_0(\mb{H}) = \max \big\{ \frac{1}{2} d(\mb{G}) \mid \mb{G}
\subseteq \mb{H} \big\}.$

A class of graphs $\CCC$ is said to be of \emph{bounded expansion} if
there is a function  $f: \nats \ra \nats$ such that for every graph
$\mb{G}$ in $\CCC$, $\grad_r(\mb{G}) \leq f(r)$.  It is known that for
every $n$, any graph with average degree $10n^2$ contains $\mb{K}_n$
as a minor (see~\cite[Theorem~7.2.1]{Die05}.  It follows immediately
that if $\CCC$ excludes $\mb{K}_n$ as a minor, it has bounded
expansion.  Indeed, the constant function $f(r) = 10n^2$ witnesses this.

Any class $\CCC$ that excludes a minor both has bounded expansion and
locally excludes minors.  However, the last two restrictions are known
to be incomparable in the sense that there are classes $\CCC$ that
locally exclude minors but are not of bounded expansion and vice versa
(see~\cite{DGK07}).  Another condition on a class $\CCC$, considered
in~\cite{ADK06} is that it has \emph{bounded degree}.  That is to say
that there is a constant $d$ such that every vertex in every graph in
$\CCC$ has degree at most $d$.  This restriction is incomparable with
the requirement that $\CCC$ excludes a minor but again, it is
immediate that any class of bounded degree both locally excludes
minors and has bounded expansion.  See~\cite{Daw07-mfcs} for a map of
these various conditions and implications between them.

\subsection{Homomorphism Preservation Theorems}

In~\cite{ADK06}, the homomorphism preservation property is established
for a number of classes of structures, based on certain combinatorial
properties that were called \emph{wide} and \emph{almost wide}
in~\cite{ADG08}.  In the following, when we talk of a class of finite
structures $\CCC$ satisfying a graph-theoretic restriction, such as
excluding a minor, we mean that the collection of Gaifman graphs
$\GGG(\Aa)$ of structures $\Aa$ in $\CCC$ satisfies the condition.

\begin{definition}\label{def:wide}
   A set of elements $B$
in a $\sigma$-structure ${\Aa}$ is \emph{$r$-scattered} if
for every pair of distinct $a,b \in B$ we have
$N_r^{\Aa}(a) \cap N_r^{\Aa}(b) = \emptyset$.

We say that a class of finite $\sigma$-structures $\mathcal{C}$ 
is \emph{wide} if for every $r$ and $m$ there exists an $N$ such that
every structure in $\mathcal{C}$ of size at least $N$ contains an
$r$-scattered set of size $m$.
\end{definition}
It is easy to see that if $\CCC$ has bounded degree, then it is wide.
Indeed, Ne\v{s}et\v{r}il and Ossona de
Mendez~\cite{NOdM08} note that for a class $\CCC$ 
that is closed under taking substructures, $\CCC$ is wide if, and only if, it has bounded degree.

\begin{definition}\label{def:almost-wide}
  A class of finite $\sigma$-structures $\mathcal{C}$ is \emph{almost
    wide with margin $k$} if for every $r$ and $m$ there exists an $N$
  such that every structure $\Aa$ with at least $N$ elements in
  $\mathcal{C}$ contains a set $B$ with at most $k$ elements such that
  $\GGG(\Aa)[A\setminus B]$ contains an $r$-scattered set of size $m$.

  We say that $\CCC$ is \emph{almost wide} if there is some $k$ such
  that it is almost wide with margin $k$.
\end{definition}
An example is the class of acyclic graphs, which is not wide (as we
have arbitrarily large trees where the distance between any two
vertices is $2$) but is almost wide with margin $1$.  More generally,
it is shown in~\cite{ADK06} that if $\CCC$ excludes $\mb{K}_n$ as a
minor, then $\CCC$ is almost wide with margin $n-2$.  A
characterisation of almost-wide classes that are closed under
subgraphs is given in~\cite{NOdM08}.

A theorem of~\cite{ADK06} shows that almost wideness, along with some
natural closure properties of a class $\CCC$ is sufficient to
guarantee the homomorphism preservation property.  
\begin{theorem}[\cite{ADK06}]\label{thm:almost-wide}
  Any class $\CCC$ of finite $\sigma$-structures that is almost wide
  and is closed under taking substructures and disjoint unions of
  structures has the homomorphism preservation property.
\end{theorem}
This is proved using a lemma of
Ajtai and Gurevich which we review in Section~\ref{sec:ajtai-g}.
Thus, as long as $\CCC$ is closed under substructures and disjoint
unions, if it has bounded degree, bounded treewidth or excludes a
minor, it has the homomorphism preservation property.  An open
question posed in~\cite{ADK06} was whether the same could be proved in
the case where $\CCC$ has \emph{bounded local treewidth}.  We will not
define this notion formally here but only note that any class of
bounded local treewidth also locally excludes minors.  Thus, by
establishing the homomorphism preservation property for classes that
locally exclude minors, we settle the open question.

\section{Quasi-Wide Classes of Structures}\label{sec:ajtai-g}
By Theorem~\ref{thm:almost-wide}, the homomorphism preservation
property holds for classes of structures which are almost wide and
closed under taking substructures and disjoint unions.  Unfortunately,
knowing that a class $\CCC$ has bounded expansion or that it locally
excludes minors is not sufficient to establish that it is almost
wide.  Indeed, it follows from the characterisation of almost-wide
classes given in~\cite{NOdM08} that there is a class of bounded
expansion and that locally excludes minors but that is not almost wide.
Our aim in this section is to show that the condition of almost
wideness can be relaxed to a weaker condition that is satisfied by the
classes we consider.  We proceed to define this condition.
\begin{definition}\label{def:quasi-wide}
  Let $f: \nats \ra \nats$ be a function.  A class of finite
  $\sigma$-structures $\mathcal{C}$ is \emph{quasi-wide with margin
    $f$} if for every $r$ and $m$ there exists an $N$ such that every
  structure $\Aa$ with at least $N$ elements in $\mathcal{C}$ contains
  a set $B$ with at most $f(r)$ elements such that
  $\GGG(\Aa)[A\setminus B]$ contains an $r$-scattered set of size $m$.

  We say that $\CCC$ is \emph{quasi-wide} if there is some $f$ such
  that $\CCC$ is quasi-wide with margin $f$.
\end{definition}
In other words, unlike in the definition of almost wide classes, the
number of elements we need to remove to guarantee a large scattered
set in a large enough structure $\Aa$ can be allowed to depend on the radius $r$ of
the neighbourhoods we consider.

Theorem~\ref{thm:almost-wide} is obtained from the following lemma
proved by Ajtai and Gurevich~\cite{AG94} and the observation that the
only constructions used in the proof involve taking substructures and
disjoint unions.  We sketch an outline of the proof below.
\begin{lemma}[Ajtai-Gurevich]\label{lem:ajtai-g}
  For any sentence $\phi$ that is preserved under homomorphisms and
  any $k\in\nats$, there are $r, m \in \nats$ such that if $\Aa$ is a
  minimal model of $\phi$ and $B \subseteq A$ is a set of its elements
  with $|B| \leq k$, then $\GGG(\Aa)[A\setminus B]$ does not contain
  an $r$-scattered set of size $m$.
\end{lemma}

Our aim here is to show that in the proof of Lemma~\ref{lem:ajtai-g},
the value of $r$ can be chosen independently of the value of $k$.
This will immediately allow us to extend Theorem~\ref{thm:almost-wide}
to quasi-wide classes of structures.  We proceed with an outline of the
proof of Ajtai and Gurevich. 

The first step in the proof is to prove it for the case when $k=0$.
Then, the general case is reduced to this special case.  So, suppose
$\phi$ is a sentence of quantifier rank $q$ that is preserved under
homomorphisms.  Let $\Sigma = \{\phi_1, \ldots, \phi_s \}$ be a
collection of basic local sentences (obtained by
Theorem~\ref{thm:gaifman}) such that $\phi$ is equivalent to a Boolean
combination of them.  For each $i$, let $t_i$ be the radius of
locality, $n_i$ the width and $\psi_i(x)$ the local condition of
$\phi_i$.  Also let $t = \max_i t_i$ and $n = \max_i n_i$.  We take $r
= 2t$ and $m = 2^s+1$.  For each $i$, we write $\theta_i(y)$ for the
following formula
$$ \exists x \big(\delta(x,y) \leq t_i \land \psi_i^{N_{t_i}(x)}(x)\big).$$

Suppose then that $\Aa$ is a model of $\phi$ that contains an
$r$-scattered set of size $m$.  We wish to show that $\Aa$ cannot be
minimal.  Suppose that $C = \{c_1,\ldots,c_m \}$ is the $r$-scattered
set.  Then, by definition $N_r^{\Aa}(c_i) \cap N_r^{\Aa}(c_j) =
\emptyset$ for $i\neq j$.  Furthermore, since $m > 2^s$, there are $i$
and $j$ with $i \neq j$ such that for all $l$, $\Aa\models
\theta_l[c_i]$ if, and only if, $\Aa \models \theta_l[c_j]$.  Let
$\Bb$ be the substructure of $\Aa$ obtained by removing some tuple
that includes $c_i$ from some relation $R$ of $\Aa$ (if there is no
such relation, then we can get a model of $\phi$ by removing the
element $c_i$, showing that $\Aa$ is not minimal in any case).
Finally, we take $\Bb_n$ to be the structure that is the disjoint union
of $n$ copies of $\Bb$ and $\Aa_n$ to be the structure that is the
disjoint union of $\Aa$ and $\Bb_n$.  Ajtai and Gurevich prove that
the structures $\Aa_n$ and $\Bb_n$ must agree on the sentence $\phi$.
Since $\phi$ is preserved under homomorphisms, and there are
homomorphisms from $\Aa$ to $\Aa_n$ and from $\Bb_n$ to $\Bb$, it
follows that if $\Aa$ is a model of $\phi$ so is $\Bb$.  Thus, since
$\Bb$ is a proper substructure of $\Aa$, the latter is not a minimal model of
$\phi$.

Note that, if $\CCC$ is a class of structures that is closed under
substructures and disjoint unions then, whenever it contains $\Aa$, it
also contains $\Bb$, $\Bb_n$ and $\Aa_n$.  Thus the above argument
showing that $\Aa$ is not minimal works in restriction to such a
class.  Note further that in the above argument establishing
Lemma~\ref{lem:ajtai-g} for $k=0$, the values of $r$ and $m$ depend on
$\phi$, but $r$ can be bounded above by $2\cdot 7^q$ where $q$ is the
quantifier rank of $\phi$, \emph{independently of the signature
  $\sigma$}.  A similar upper bound for $m$ is not obtained as this
depends on the number of inequivalent basic local sentences of a given
quantifier rank and locality radius that can be expressed and this, in
turn, depends on the signature.

The proof of Lemma~\ref{lem:ajtai-g} by Ajtai and Gurevich then
proceeds to reduce the case $k>0$ to the case $k=0$ by means of the
construction of what they call \emph{plebeian companions}.  That is, for every structure $\Aa$ and a tuple of elements $\tup{a} = (a_1,\ldots,a_k)$ from $\Aa$ we define a structure $p\Aa_{\tup{a}}$ called the \emph{plebeian companion} of $\Aa$.  This is a structure over a richer vocabulary than $\Aa$ and has the property that $\GGG(p\Aa_{\tup{a}}) \iso \GGG(\Aa)[A\setminus \tup{a}]$.  In particular, $p\Aa_{\tup{a}}$ contains an $r$-scattered set of $m$ elements if, and only if, removing the elements $a_1,\ldots,a_k$ from $\GGG(\Aa)$ creates such a set. 
Furthermore, Ajtai and Gurevich give a translation that takes a formula $\phi$ in the signature $\tau$ of $\Aa$ to a formula $\hat{\phi}$ in the signature $\tau'$ of $p\Aa_{\tup{a}}$ so that $\Aa\models \phi$ if, and only if, $p\Aa_{\tup{a}}\models \hat{\phi}$ and $\hat{\phi}$ is preserved under homomorphisms if $\phi$ is.  This then allows us to deduce Lemma~\ref{lem:ajtai-g} since if $\Aa$ is a model of $\phi$ and $B = \{a_1,\ldots,a_k\}$ a set of elements such that $ \GGG(\Aa)[A\setminus B]$ contains an $r$-scattered set of $m$ elements, we can note (from the case $k=0$) that $p\Aa_{\tup{a}}$ is not a minimal model of $\hat{\phi}$.  Moreover, from a proper submodel of the latter we can reconstruct a proper substructure of $\Aa$ that is a model of $\phi$ establishing that $\Aa$ is not minimal.

Our aim here is to show that in the translation of $\phi$ to $\hat{\phi}$, while the signature of $\hat{\phi}$ depends on the value of $k$, the quantifier rank is actually the same as that of $\phi$.  To this end, we give the translation in detail.

Fix a structure $\Aa$ in a relational signature $\tau$ and a tuple of elements $a_1,\ldots,a_k$ from $A$.  The signature $\tau'$ contains all the relation symbols in $\tau$.  In addition, for each relation symbol $R$ of arity $r$ in $\tau$ and each non-empty partial function $\mu: \{1,\ldots,r\} \pfunc \{a_1,\ldots,a_k\}$, $\tau'$ contains a new relation symbol $R_{\mu}$ whose arity is $r-j$ where $j$ is the number of elements of $\{1,\ldots,r\}$ on which $\mu$ is defined.  In particular, if $\mu$ is total, $r=j$ and $R_{\mu}$ is then a $0$-ary relation symbol. That is to say, it is a Boolean symbol that is interpreted as either true or false in any $\tau'$-structure.

The universe of $p\Aa_{\tup{a}}$ is obtained from that of $\Aa$ by
excluding the elements $a_1,\ldots,a_k$.  For each relation symbol $R$
in $\tau$, the interpretation of $R$ in $p\Aa_{\tup{a}}$ is the
restriction of $R^{\Aa}$ to the universe of $p\Aa_{\tup{a}}$. To
define the interpretation of $R_{\mu}$, let $\mb{b}$ be an $r-j$ tuple of
elements from $p\Aa_{\tup{a}}$.  Let $\mb{b}'$ be the $r$-tuple of
elements of $\mb{A}$ obtained from $\mb{b}$ by inserting in position
$i$ the element $\mu(i)$.  We say that $\mb{b} \in R_{\mu}^{p\Aa_{\tup{a}}}$
if, and only if, $\mb{b}' \in R^{\Aa}$.  In the special case that
$R_m$ is $0$-ary, we say that it is interpreted as true if, and only if,
the unique empty tuple is in $R_{\mu}$ by the above rule.

To describe the translation of $\phi$ to $\hat{\phi}$, we consider an
expansion of the signature $\tau$ with constants for the elements
$a_1,\ldots,a_k$ (we do not distinguish between the elements and the
constants that name them).  Note that these constants appear neither
in $\phi$ nor in $\hat{\phi}$ but they are useful in the inductive
definition of the translation.  So we proceed to define the
translation by induction on the structure of a formula $\phi$ in the
expanded signature.
\begin{itemize}
\item If $\phi$ is the atomic formula $R\tup{t}$ and the tuple of
  terms $\tup{t}$ does not contain any of the constants
  $a_1,\ldots,a_k$, then $\hat{\phi} \defeq \phi$.
\item If $\phi$ is the atomic formula $R\tup{t}$ and $\tup{t}$
  contains constants from $a_1,\ldots,a_k$, let $\mu$ be the partial
  function that takes $i$ to the constant appearing in position $i$ of
  $\tup{t}$.  Also, let $\tup{t}'$ be the tuple of variables obtained
  from $\tup{t}$ by removing the constants.  Then $\hat{\phi} \defeq
  R_{\mu}\tup{t}'$.
\item If $\phi$ is $\neg \psi$, then $\hat{\phi} \defeq \neg \hat{\psi}$ and if $\phi$ is $\psi_1 \land \psi_2$ then $\hat{\phi} \defeq \hat{\psi_1} \land \hat{\psi_2}$.
\item If $\phi$ is $\exists x \psi$ then $\hat{\phi} \defeq  \exists x \hat{\psi} \vee \bigvee_{i=1}^k \hat{\psi[x/a_i]}$.
\end{itemize}

It is clear from this translation that, while the signature of $\hat{\phi}$ depends on the value of $k$, its quantifier rank is the same as the quantifier rank of $\phi$.  Combining this with the fact that in the proof of Lemma~\ref{lem:ajtai-g} for the case $k=0$, we could bound the value of $r$ by $2\cdot 7^q$ independently of the signature of $\phi$, gives us the following strengthening of Lemma~\ref{lem:ajtai-g}.
\begin{lemma}\label{lem:stronger-ag}
  For any sentence $\phi$ of quantifier rank $q$ that is preserved under homomorphisms and
  any $k\in\nats$, there is an $m \in \nats$ such that if $\Aa$ is a
  minimal model of $\phi$ and $B \subseteq A$ is a set of its elements
  with $|B| \leq k$, then $\GGG(\Aa)[A\setminus B]$ does not contain
  a $2\cdot 7^q$-scattered set of size $m$.
\end{lemma}

Since, by the observation in~\cite{ADK06}, this holds relativised to any class of structures $\CCC$ closed under substructures and disjoint unions, we obtain the following theorem.
\begin{theorem}\label{thm:quasi-wide}
  Any class $\CCC$ of structures that is quasi-wide and closed under substructures and disjoint unions has the homomorphism preservation property.
\end{theorem}
\begin{proof}
Let $f : \nats \ra \nats$ be such that $\CCC$ is quasi-wide with margin $f$.
  Let $\phi$ be a sentence that is preserved under homomorphisms on $\CCC$.  By Lemma~\ref{lem:minimal} it suffices to prove that there is an $N$ such that no minimal model of $\phi$ in $\CCC$ has more than $N$ elements.

Write $q$ for the quantifier rank of $\phi$, let $r \defeq 2\cdot 7^q$ and let $k\defeq f(r)$.  Lemma~\ref{lem:stronger-ag} then gives us an $m$ such that in any minimal model of $\phi$ the removal of $k$ elements cannot create an $r$-scattered set of size $m$.  However, Definition~\ref{def:quasi-wide} ensures that there is an $N$ such that any structure in $\CCC$ with more than $N$ elements contains $k$ elements whose removal creates just such a scattered set.  We conclude that no minimal model of $\phi$ contains more than $N$ elements.
\end{proof}

\section{Bounded Expansion and Locally Excluded Minors}\label{sec:classes}

Our aim in this section is to show that classes of graphs that locally
exclude minors or that have bounded expansion are quasi-wide.  The
proof of this is an adaptation of the proof from~\cite{ADK06} that
classes of structures that exclude a minor are almost wide.  To be
precise, it is shown there that the following holds.
\begin{theorem}[\cite{ADK06}]\label{thm:minor}
  For any $k,r,m \in \nats$ there is an $N \in \nats$ such that if $\mb{G} = (V,E)$ is a graph with more than $N$ vertices then
\begin{enumerate}
\item either $\mb{K}_k \minor \mb{G}$; or
\item there is a set $B\subseteq V$ with $|B|\leq k-2$ such that $\mb{G}[V\setminus B]$ contains an $r$-scattered set of size $m$.
\end{enumerate}
\end{theorem}

The proof of Theorem~\ref{thm:minor} is a Ramsey-theoretic argument that proceeds by starting with a set $S\subseteq V$ with $N$ elements and constructing two sequences of sets: $S =: S_0 \supseteq S_1 \supseteq \cdots \supseteq S_{r}$ and $\emptyset =: B_0 \subseteq B_1 \subseteq \cdots \subseteq B_{r}$ such that for each $x,y \in S_i$ we have $N_i^{\mb{G}[V\setminus B_i]}(x) \cap N_i^{\mb{G}[V\setminus B_i]}(y) = \emptyset$.  If $\mb{K}_k \not\minor \mb{G}$ then we can carry the construction through for $r$ stages and $|S_r| \geq m$ and $|B_r| \leq k-2$.  If the construction fails at some stage $i \leq r$, it is because we have found that $\mb{K}_k$ is a minor of $\mb{G}$ and this can happen in one of three ways.
\begin{itemize}
\item We find that there are $s_1,\ldots,s_k \in S_i$ such that for each $1\leq j < l \leq k$, there is an edge between some vertex in $N_i^{\mb{G}[V\setminus B_i]}(s_j)$ and $N_i^{\mb{G}[V\setminus B_i]}(s_l)$.  In this case, we can take the collection of sets $N_i^{\mb{G}[V\setminus B_i]}(s_j)$ for $1 \leq j \leq k$ as branch sets.
\item We find that there are $s_1,\ldots,s_k \in S_i$ such that there are distinct vertices $x_{jl}$ for each $1\leq j < l \leq k$, where each $x_{jl}$ is a neighbour to some vertex in $N_i^{\mb{G}[V\setminus B_i]}(s_j)$ and to some vertex in $N_i^{\mb{G}[V\setminus B_i]}(s_l)$.  In this case, we find that $\mb{K}_k$ is a minor of $\mb{G}$ by taking as branch sets $N_i^{\mb{G}[V\setminus B_i]}(s_j) \cup \{ x_{jl} \mid j< l \}$ for  $1 \leq j \leq k$.
\item We find $s_1,\ldots,s_{k-1} \in S_i$ and vertices $x_1,\ldots,x_{k-1}$ such that $x_j$  has edges connecting it to each of the sets $N_{i}^{\mb{G}[V\setminus B_i]}(s_j)$.  Thus, $\mb{K}_k$ is found as a minor of $\mb{G}$ by taking as branch sets: $N_{i}^{\mb{G}[V \setminus B_i]}(s_j) \cup \{x_j\}$ for $1\leq j \leq k-2$ along with $N_{i}^{\mb{G}[V\setminus B_i]}(s_{k-1})$ and $\{x_{k-1}\}$.
\end{itemize}

The point of this brief recapitulation of the proof is to note that 
when $\mb{K}_k$ is found as a minor of $\mb{G}$ in case~(1) of the
theorem, the branch sets have radius at most $r+1$.  Thus, we actually
obtain the following stronger theorem.
\begin{theorem}\label{thm:shallow-minor}
  For any $k,r,m \in \nats$ there is an $N \in \nats$ such that if $\mb{G} = (V,E)$ is a graph with more than $N$ vertices then
\begin{enumerate}
\item either $\mb{K}_k \minor_{r+1} \mb{G}$; or
\item there is a set $B\subseteq V$ with $|B|\leq k-2$ such that $\mb{G}[V\setminus B]$ contains an $r$-scattered set of size $m$.
\end{enumerate}
\end{theorem}

We write $N(k,r,m)$ for the value of $N$ obtained from Theorem~\ref{thm:shallow-minor} for given $k,r$ and $m$.

The following result now follows immediately.

\begin{theorem}\label{thm:bounded-expansion}
  Any class of graphs of bounded expansion is quasi-wide. 
\end{theorem}
\begin{proof}
Suppose that $\CCC$ is a class of graphs of bounded expansion and let $f$ be a function such that for any graph $\mb{G}$ in $\CCC$, $\grad_r(\mb{G}) \leq f(r)$.  Let $k(r) := 2f(r+1)+2$.  Note that $$\frac{|E^{\mb{K}_{k(r)}}|}{|V^{\mb{K}_{k(r)}}|} = \frac{k(r)-1}{2} > f(r+1)$$
and therefore, by the definition of bounded expansion, $\mb{K}_{k(r)} \not\minor_{r+1} \mb{G}$ for any graph $\mb{G}$ in $\CCC$.  Thus, by Theorem~\ref{thm:shallow-minor}, if $\mb{G}$ has more than $N(k(r),r,m)$ vertices, it contains a set $B$ with at most $k(r)-2$ vertices such that $\mb{G}[V^{\mb{G}}\setminus B]$ contains an $r$-scattered set of size $m$.  Thus, $\CCC$ is quasi-wide with margin $k(r)-2$.
\end{proof}

We now consider the case of classes with locally excluded minors.  It is useful to first derive a straightforward corollary to Theorem~\ref{thm:shallow-minor}.
\begin{corollary}\label{cor:local-minors}
If $\mb{G} = (V,E)$ is a graph with more than $N(k,r,m)$ vertices then
\begin{enumerate}
\item either there is a $v \in V$ such that $\mb{K}_k \minor N_{3r+4}^{\mb{G}}(v)$; or
\item there is a set $B\subseteq V$ with $|B|\leq k-2$ such that $\mb{G}[V\setminus B]$ contains an $r$-scattered set of size $m$.
\end{enumerate}
\end{corollary}
\begin{proof}
  Suppose condition~(2) fails.  Then, by
  Theorem~\ref{thm:shallow-minor}, we have $\mb{K}_k \minor_{r+1}
  \mb{G}$.  Let $\mb{H}_1,\ldots,\mb{H}_k$ be the branch sets that
  witness this and let $v_1,\ldots,v_k$ be vertices such that
  $\mb{H}_i \subseteq N_{r+1}^{\mb{G}}(v_i)$.  Then, for any $i$ and
  any vertex $u$ in $\mb{H}_i$ there is a path of length at most
  $3r+4$ from $v_1$ to $u$.  This is because there is an edge between
  some vertex $w$ in $\mb{H}_1$ and a vertex $w'$ in $\mb{H}_i$.
  Moreover, there is a path of length at most $r+1$ from $v_1$ to $w$
  and since $u,w' \in N_{r+1}^{\mb{G}}(v_i)$, there is a path of
  length at most $2r+2$ from $w'$ to $u$.  Thus, $\bigcup_{i=1}^k
  \mb{H}_i \subseteq N_{3r+4}^{\mb{G}}(v_1)$ and hence $\mb{K}_k
  \minor N_{3r+4}^{\mb{G}}(v_1)$.
\end{proof}

\begin{theorem}\label{thm:local-minors}
  Any class of graphs that locally excludes minors is quasi-wide.
\end{theorem}
\begin{proof}
Suppose $\CCC$ is a class of graphs that locally excludes minors.  In particular, let $f$ be a function such that for any $r$, $\mb{K}_{f(r)} \not\minor N^{\mb{G}}_r(v)$ for any graph $\mb{G}$ in $\CCC$ and any vertex $v$ of $\mb{G}$.

Now, for any $r$, let $k(r) \defeq f(3r+4)$.  By definition, for any graph $\mb{G}$ in $\CCC$ and any vertex $v$ of $\mb{G}$, $K_{k(r)} \not\minor N_{3r+4}^{\mb{G}}(v)$.  Thus, by Corollary~\ref{cor:local-minors}, if $\mb{G}$ has more than $N(k(r),r,m)$ vertices, it contains a set $B$ with at most $k(r)-2$ vertices such that $\mb{G}[V^{\mb{G}}\setminus B]$ contains an $r$-scattered set of size $m$.  Thus, $\CCC$ is quasi-wide with margin $k(r)-2$.
\end{proof}

We can now state the main results of the paper.
\begin{theorem}
  Any class $\CCC$ of finite structures that has bounded expansion and is closed under taking substructures and disjoint unions has the homomorphism preservation property.
\end{theorem}
\begin{proof}
  Immediate from Theorem~\ref{thm:quasi-wide} and Theorem~\ref{thm:bounded-expansion}.
\end{proof}

\begin{theorem}
  Any class $\CCC$ of finite structures that locally excludes minors and is closed under taking substructures and disjoint unions has the homomorphism preservation property.
\end{theorem}
\begin{proof}
  Immediate from Theorem~\ref{thm:quasi-wide} and Theorem~\ref{thm:local-minors}.
\end{proof}
\section{Failure of Preservation}\label{sec:counterexample}
In this section we give an example of a class of structures $\SSS$
which is closed under substructures and disjoint unions but does not
have the homomorphism preservation property.

The class $\SSS$ is over a signature $\tau$ with two binary relations
$O$ and $S$ and one unary relation $P$.  For any $n \in \nats$, let
$L_n$ be the $\tau$-structure over the universe $\{1,\ldots,n\}$ in
which $O$ is interpreted as the usual linear order, i.e.\ $O(i,j)$
just in case $i<j$; $S$ is the successor relation: $S(i,j)$ just in
case $j=i+1$; and $P$ is interpreted by the set $\{1,n\}$ containing
the two endpoints.  Let $\LLL$ be the class of structures isomorphic
to $L_n$ for some $n$.  Then $\SSS$ is the closure of $\LLL$ under
substructures and disjoint unions.  Note that every structure $\Aa$ in
$\SSS$ is isomorphic to the disjoint union of a collection $A_1,\ldots,A_s$ of
structures, each of which is a substructure of some $L_n$.

We begin with some observations about structures in $\SSS$.  

\begin{lemma}\label{lem:single-order}
If $\Aa$ is
a structure such that $\Aa \subseteq L_m$ for some $m$ and there is a
homomorphism $h: L_n \ra \Aa$ for some $n\geq 2$, then $L_n \iso \Aa \iso L_m$.
\end{lemma}
\begin{proof}
Note that, by definition of the structures $L_m$, if $O(a,b)$ for two
elements $a,b$ of $\Aa$ then $a\neq b$.
Since $L_n$ contains two elements $1,n$ in the set $P$ with $O(1,n)$
we conclude that $\Aa$ contains both endpoints of $L_m$ and they are
both in the set $P^{\Aa}$.  Furthermore, $L_n$
contains an $S$-path from $1$ to $n$.  The image of this path under
$h$ must be an $S$-path between the end points of $L_m$ and we
conclude that $m=n$ and $h$ is the identity map.
Finally, suppose that for some $i,j$ in $L_m$ with $i<j$, the pair
$(i,j)$ is not in $O^{\Aa}$.  But then, since $(i,j) \in O^{L_n}$ and $h$ is the identity,
$h$ is not a homomorphism.  We conclude that $\Aa \iso L_n$.
\end{proof}

Say that a structure $\Aa \in \SSS$ \emph{contains a complete order}
if there is some $n\geq 2$ such that $L_n \subseteq \Aa$.
\begin{lemma}\label{lem:contains-order}
  If $\Aa$ and $\Bb$ in $\SSS$ are such that $\Aa$ contains a complete
  order and there is a homomorphism $h: \Aa \ra \Bb$, then $\Bb$
  contains a complete order.
\end{lemma}
\begin{proof}
  Suppose $L_n \subseteq \Aa$ and $\Bb = B_1 \oplus \cdots \oplus B_s$
  where for each $i$, $B_i \subseteq L_{m_i}$ for some $m_i$.  Since the
  $B_i$ are pairwise disjoint and $L_n$ is connected there is some $i$
  such that $h(L_n) \subseteq B_i$.  But then, by
  Lemma~\ref{lem:single-order}, $B_i \iso L_n$ and so $\Bb$ contains a
  complete order.
\end{proof}

Our aim now is to construct a first-order sentence that defines those
structures in $\SSS$ that contain a complete order.

We write $x \leq y$ as an abbreviation for the formula $O(x,y) \vee
x=y$.  Let $\beta(x,y,z)$ denote the formula $x \leq z \land z \leq y$
and let $\lambda(x,y)$ denote the formula that asserts that $O(x,y)$
and that $\leq$ linearly orders the set of elements $\{z \mid x\leq z
\mbox{ and } z\leq y\}$.  That is, $\lambda(x,y)$ is the formula:
$$ O(x,y) \land \forall z_1 \forall z_2 (\beta(x,y,z_1) \land \beta(x,y,z_2))
\ra (z_1 \leq z_2 \vee z_2 \leq z_1). $$
Let $\nu(z_1,z_2)$ denote the formula $O(z_1,z_2) \land \forall w
\neg(O(z_1,w) \land O(w,z_2))$.  In words, $\nu(z_1,z_2)$ defines the
pairs of elements in the relation $O$ with nothing in between them.
We are now ready to define the sentence $\phi$:
$$ 
\begin{array}{l@{}l}
\exists x \exists y (P(x) & \land P(y)  \land \lambda(x,y) \land \\
  & \land \forall z_2 \forall z_2 (\beta(x,y,z_1) \land \beta(x,y,z_2)
  \land \nu(z_1,z_2) ) \ra S(z_1,z_2)).
\end{array}
$$
That is, $\phi$ asserts that there exist two elements $x$ and $y$ in
the relation $P$ such that the set $\{z \mid x\leq z \mbox{ and }
z\leq y\}$ is linearly ordered by $O$ and any two successive elements
in that linear order are related by $S$.
\begin{lemma}\label{lem:def-phi}
  For any $\Aa$ in $\SSS$, $\Aa \models \phi$ if, and only if, $\Aa$
  contains a complete order.
\end{lemma}
\begin{proof}
  It is clear that if $L_n \subseteq \Aa$, then $\Aa \models \phi$ with the endpoints of $L_n$ being witnesses to the outer existential quantifiers.
  For the converse, suppose that $\Aa \models \phi$ and $a$ and $b$
  are elements witnessing the outer existential quantifiers.  By the
  facts $P(a)$, $P(b)$ and $O(a,b)$ we know that there is an $A_i
  \subseteq \Aa$ and an $n$ such that $A_i \subseteq L_n$ with $a,b$
  being the endpoints of $L_n$.  The sentence $\phi$ then guarantees
  that $A_i$ contains all elements of $L_n$ and all tuples in the
  relations.   Thus $A_i \iso L_n$ and so $\Aa$ contains a complete
  order. 
\end{proof}

\begin{lemma}
  The formula $\phi$ is preserved under homomorphisms on the class $\SSS$.
\end{lemma}
\begin{proof}
   Immediate from Lemmas~\ref{lem:contains-order}
and~\ref{lem:def-phi}. 
\end{proof}

\begin{lemma}
  There is no existential positive formula equivalent to $\phi$ on $\SSS$.
\end{lemma}
\begin{proof}
  By Lemma~\ref{lem:minimal}, it suffices to show that $\phi$ has
  infinitely many minimal models in $\SSS$.  But this is immediate as
  for every $n\geq 2$, $L_n$ is a model of $\phi$ but no proper substructure of
  $L_n$ is a model of $\phi$.
\end{proof}

It is worth remarking that the collection of Gaifman graphs of structures in $\SSS$ is the class of all graphs and hence is certainly not quasi-wide.

\section{Conclusions}\label{sec:conclusion}
When $\CCC$ is a class of finite structures, there are essentially two
methods known for showing that it has the homomorphism preservation
property.  One is the method used by Rossman to establish the property
for the class of all finite structures, based on constructing
sufficiently saturated structures.  This method works on any
class closed under co-retracts.  The other, quite distinct method,
developed by Atserias et al., is based on the density of minimal
models and works for classes of sparse structures, i.e.\ classes in
which any sufficiently large structure is guaranteed not to be dense.
In the present paper, we have pushed the latter method further and
established the homomorphism preservation property for a richer
collection of classes.  None of these classes, it appears, is closed
under the kind of saturation construction used by Rossman and
therefore those methods would not apply.

\end{document}